\documentclass[12pt,english]{article}
\usepackage[margin=1.5in]{geometry}

\usepackage{mdwlist}
\usepackage{enumerate}
\usepackage{amssymb,amsbsy,latexsym}
\usepackage{amsmath}
\usepackage{graphics, subfigure, float}
\usepackage{fp, calc}
\usepackage{hyperref}
\usepackage{url}

\usepackage{bm}

\usepackage{amscd,amsthm}

\usepackage{pst-all}
\usepackage{pstricks-add}
\usepackage{pst-func}
\newpsobject{showgrid}{psgrid}{subgriddiv=1,griddots=10,gridlabels=6pt}
\usepackage{verbatim, comment}
\usepackage{datetime}

\newtheoremstyle{theorem}{1em}{1em}{\slshape}{0pt}{\bfseries}{.}{ }{}
\theoremstyle{theorem}
\newtheorem{theorem}{Theorem}

\newtheorem*{theorem*}{Theorem}

\newtheorem*{claim*}{Claim}
\newtheorem*{lemma*}{Lemma}
\theoremstyle{remark}
\newtheorem{remark}{Remark}
\newtheorem*{remark*}{Remark}

\providecommand{\setN}{\mathbb{N}}
\providecommand{\setZ}{\mathbb{Z}}

\providecommand{\setR}{\mathbb{R}}

        \def\drawRect#1#2#3#4#5{
           \FPeval{\x2}{(#2) + (#4)} 
           \FPeval{\y2}{(#3) + (#5)} 
           \pspolygon[#1](#2,#3)(\x2,#3)(\x2,\y2)(#2,\y2)
        }

\usepackage[displaymath,textmath,graphics, subfigure, floats]{preview} 
\PreviewEnvironment{center} 
\PreviewEnvironment{pspicture} 

\title{A note on the $\Sigma_2^P$-completeness of the Frobenius number}
\author{Thomas Rothvoss\thanks{University of Washington, Seattle. Email: {\tt rothvoss@uw.edu}. Supported by NSF grant 2318620 \emph{AF: SMALL: The Geometry of Integer Programming and Lattices}.}}
\date{}

\begin{document}

\maketitle

\begin{abstract}
  Given a finite set $A$ of natural numbers whose greatest common divisor is one, the \emph{Frobenius number} $g(A)$ is the largest integer
  that is not a non-negative integer combination of the numbers in $A$. In a 2016 preprint,
  Matsubara states that given $A$ and $k$, deciding if $g(A) \geq k$ is $\Sigma_2^P$-complete.
  A decade has passed since without peer-reviewed publication of this result. At the same
  time, the community has found it difficult to verify this result. In this note, we give
  a write-up of the completeness proof based on Matsubara (2016).
\end{abstract}

\begin{remark}
\emph{The author used OpenAI's GPT 5.6 Sol to decode the proof of \cite{matsubara2016computationalcomplexityfrobeniusproblem}.  The writing here is 100\% due to the author.
  The author takes no  intellectual credit; the sole purpose of this note is to clarify the complexity status of computing the Frobenius number.}
\end{remark}

\section{Introduction}

Given a set $A = \{ a_1,\ldots,a_n\}$ of integers $a_1,\ldots,a_n \in \setN$,
we say that $x \in \setN$ is \emph{$A$-representable} if there are coefficients $y \in \setZ_{\geq 0}^n$ so that $x = \sum_{i=1}^n a_iy_i$. By Schur's Theorem (as reported in \cite{Brauer1942}; see \cite{Beihoffer_Hendry_Nijenhuis_Wagon_2005} for a more contemporary account) we know that for any set $A \subseteq \setN$ with $\gcd(A) := \gcd(a_1,\ldots,a_n) = 1$,
all large enough numbers are $A$-representable\footnote{In fact, Schur gave an explicit bound: if $2 \leq a_1 < \ldots < a_n$, then all integers $x \geq (a_1-1) \cdot (a_n-1)$ are $A$-representable.}
This motivates
the definition of the \emph{Frobenius number} $g(A)$ which is the largest integer that is \emph{not} $A$-representable\footnote{For definiteness, we set $g(A) := \infty$ if $\textrm{gcd}(A) > 1$.}. For example one can verify that $g(\{5,6\}) = 19$ as $19$ is not $\{5,6\}$-representable but every larger integer is $\{ 5,6\}$-representable. The Frobenius number has an intuitive interpretation as the \emph{coin problem}: given a coin system with $n$ different denominations $a_1,\ldots,a_n$, find the largest amount that cannot be paid exactly using such coins.  

In a landmark paper, Kannan~\cite{KannanFSTTCS1989,KannanCombinatorica92} proved that for any fixed $n$, there is a polynomial time algorithm to compute $g(A)$ for sets with $|A| =n$. In fact, Kannan solves the more general
problem of computing the covering radius $\mu(\setZ^n,K)$ for any rational polytope $K \subseteq \setR^n$ in polynomial time as long as $n$ is constant. However, this note focuses on the complexity-theoretic aspects
of the Frobenius number. For that purpose we define a decision variant of the Frobenius problem
in the form of the language
\[
  \texttt{FROB} := \{ (A,k) : g(A) \geq k \}
\]
Ram{\'i}rez-Alfons{\'i}n~\cite{RamrezAlfonsn1996ComplexityOT} proved that $\tt{FROB}$ is {\bf NP}-hard unter Turing reductions\footnote{In a Turing reduction one is allowed to use polynomially many oracle calls of $\texttt{FROB}$.}. This result was followed by Matsubara giving the exact classification: 
\begin{theorem}[Matsubara~\cite{matsubara2016computationalcomplexityfrobeniusproblem}]
${\tt FROB}$ is $\Sigma_2^P$-complete.
\end{theorem}
However, at the time of this writing, no formal publication of \cite{matsubara2016computationalcomplexityfrobeniusproblem} is available and the writing makes verification difficult\footnote{See e.g. the discussion in \url{https://cstheory.stackexchange.com/questions/42499/computational-complexity-of-the-frobenius-problem}.}
The purpose of this note is to clarify the proof of \cite{matsubara2016computationalcomplexityfrobeniusproblem}.

\section{Preliminaries}

We recap the necessary complexity theory background. We recommend the standard textbook by Arora and Barak~\cite{AroraBarak2009} for more details.
For $n \in \setN$, we use $[n] = \{1,\ldots,n\}$. A \emph{language} is any set of strings $L \subseteq \{ 0,1\}^*$.
The complexity class $\Sigma_2^P$ is the set of languages $L \subseteq \{ 0,1\}^*$ that can be written in the form
\[
 L = \big\{ x \in \{ 0,1\}^* \mid \exists u \in \{ 0,1\}^{p_1(|x|)} \; \forall v \in \{ 0,1\}^{p_2(|x|)} : M(x,u,v) = 1\big\}
\]
where $p_1$ and $p_2$ are polynomials and $M$ is a deterministic polynomial time algorithm. Analogously, $\Pi_2^P$ is the set of languages
\[
 L = \big\{ x \in \{ 0,1\}^* \mid \forall u \in \{ 0,1\}^{p_1(|x|)} \; \exists v \in \{ 0,1\}^{p_2(|x|)} : M(x,u,v) = 1\big\}
\]
We note that $\Pi_2^P$ is the complementary class to $\Sigma_2^P$ in the sense that $L \in \Sigma_2^P \Leftrightarrow \bar{L} \in \Pi_2^P$. Both classes generalize to $\Sigma_k^P$ and $\Pi_k^P$ which have $k$ alternating quantifiers and the union $\bigcup_{k \geq 0} \Sigma_k^P = \bigcup_{k \geq 0} \Pi_k^P$ forms the \emph{polynomial hierarchy} ${\bf PH}$. 

For two languages $L_1,L_2 \subseteq \{ 0,1\}^*$ we write $L_1 \leq_p L_2$ if there is a polynomial-time computable map $f : \{ 0,1\}^* \to \{ 0,1\}^*$ so that
\[
 x \in L_1  \quad \Longleftrightarrow \quad f(x) \in L_2 \quad \forall x \in \{ 0,1\}^*
\]
This is also called a \emph{Karp-reduction}. 
We say that a language $L \subseteq \{ 0,1\}^*$ is \emph{$\Sigma_2^P$-hard} if for every $L' \in \Sigma_2^P$ one has $L' \leq_p L$.
Moreover $L$ is $\Sigma_2^P$-complete if (i) $L \in \Sigma_2^P$ and (ii) $L$ is $\Sigma_2^P$-hard. 
Analogously we define $\Pi_2^P$-completeness. The relation $\leq_p$ is transitive and so in order to
prove say $\Pi_2^P$-hardness of $L$, it suffices to take any $\Pi_2^P$-hard problem $L'$ and
show that $L' \leq_p L$. A suitable hard problem will be the following:
\begin{quote}
\emph{Quantified 3-dimensional matching} (${\tt Q3DM}$). Given disjoint sets $V_1,V_2,V_3$ with $n = |V_1| = |V_2| = |V_3|$ and sets of hyperedges $M_1,M_2 \subseteq V_1 \times V_2 \times V_3$ with $M_1 \cap M_2 = \emptyset$. Decide if for all $S_1 \subseteq M_1$ there is a set $S_2 \subseteq M_2$ so that the union $S_1 \cup S_2$ is a perfect matching.
\end{quote}
Here, an edge $e \in M_1 \cup M_2$ is actually a hypergraph edge with $|e| = 3$. As usually,
a set $S \subseteq M_1 \cup M_2$ is called a \emph{perfect matching} if for all $v \in V_1 \cup V_2 \cup V_3$
there is precisely one edge $e \in S$ with $v \in e$. Spelling out the order of the quantifiers we can write
\[
 {\tt Q3DM} = \big\{ (M_1,M_2) : \forall S_1 \subseteq M_1 \; \exists S_2 \subseteq M_2: S_1 \cup S_2\textrm{ is a perfect matching} \big\}
\]
We rely on the following classical result by McLoughlin:
\begin{theorem}[McLoughlin~\cite{McLoughlinHardness1984}]
  The language ${\tt Q3DM}$ is $\Pi_2^P$-complete.
\end{theorem}

\section{The main result}

Since ${\tt Q3DM}$ is in $\Pi_2^P$ rather than $\Sigma_2^P$, we work with the complementary
Frobenius language\footnote{Technically speaking this is not the precise set-theoretic complement of ${\tt FROB}$ as neither contains strings that are not syntactically valid instances. But that difference is a polynomial time recognizable subset and hence it can be ignore for complexity purposes. This is standard procedure, see e.g. \cite{AroraBarak2009}.}
\begin{equation} \label{eq:QuantorRepForFrobenius}
 \overline{\tt FROB} = \big\{ (A,k): g(A) < k\big\} = \Big\{ (A,k) \mid \forall x \geq k: \exists y \in \setZ_{\geq 0}^n : \sum_{i=1}^n a_iy_i = x \Big\}
\end{equation}
As discussed above, ${\tt FROB}$ is $\Sigma_2^P$-complete if and only if $\overline{{\tt FROB}}$ is $\Pi_2^P$-complete. Hence it suffices to show the following:
\begin{theorem}
$\overline{{\tt FROB}}$ is $\Pi_2^P$-complete.
\end{theorem}
\begin{proof}
  We know that  $\overline{{\tt FROB}} \in \Pi_2^P$ from the characterization in~\eqref{eq:QuantorRepForFrobenius} together with the fact that $\textrm{gcd}(A)$ can be computed in polynomial time and the observation that Schur's Theorem gives an explicit upper bound on the Frobenius number whose encoding length is polynomial in the encoding length of $A$.

  In order to show $\Pi_2^P$-hardness we will prove that
  ${\tt Q3DM} \leq_p \overline{{\tt FROB}}$. Fix an instance $(M_1,M_2)$ for ${\tt Q3DM}$ with vertex sets $V_1,V_2,V_3$. We may assume that $|V_1|=|V_2|=|V_3|=n$,  $M_1$ is a matching and $M_2$ is non-empty since outwise it is trivial to decide if $(M_1,M_2) \in {\tt Q3DM}$. We write $V_r = \{ v_{r,1},\ldots,v_{r,n}\}$ for $r \in \{ 1,2,3\}$ and $V := V_1 \dot{\cup} V_2 \dot{\cup} V_3$.
  For a vertex $v_{r,j} \in V$ with $r \in [3]$, $j \in [n]$, we define an index $p(r,j) := (r-1) \cdot n + (j-1)$. The $p$-index simply maps the $3n$ vertices in $V$ bijectively to $\{ 0,\ldots,3n-1\}$.
  Suppose $M_1 = \{ e_1,\ldots,e_m \}$ are the edges in the matching.
  We give the edge $e_i$ the index $s(i) := 3n + (i-1)$. That means the function $s$ maps the edges in  $M_1$ bijectively to the positions $\{ 3n,\ldots,3n+m-1\}$.

  Now we are ready to define the Frobenius instance. We will represent numbers using the base $b := n+1$.
  \begin{eqnarray*}
    A_1 &:=& \Big\{ b^{3n+m} + t b^{s(i)} + \sum_{r=1}^3 d_{r}b^{p(r,j_{r})} \mid e_i=(v_{1,j_1},v_{2,j_2},v_{3,j_3}) \in M_1;  t \in [n], d_{r} \in \{ 0,\ldots,n\} \Big\} \\
    A_2 &:=& \Big\{ b^{3n+m} + \sum_{r=1}^3 d_{r}b^{p(r,j_r)} \mid (v_{1,j_1},v_{2,j_2},v_{3,j_3}) \in M_2; d_{r} \in \{ 0,\ldots,n\} \Big\} \\
             A &:=& A_1 \cup A_2\\
    k &:=& nb^{3n+m}
  \end{eqnarray*}
We may assume that $m \leq n^3$ and so $|A|$ and the encoding length of all numbers in $A$ are at most polynomial in $n$. 
  To understand the construction it will be useful to consider the $b$-ary representation in which
  each number $a \in A_1$ has at most $5$ non-zero digits while each $a \in A_2$ has at most 4 non-zero digits, see also Figure~\ref{fig:bAryRepOfA}.
    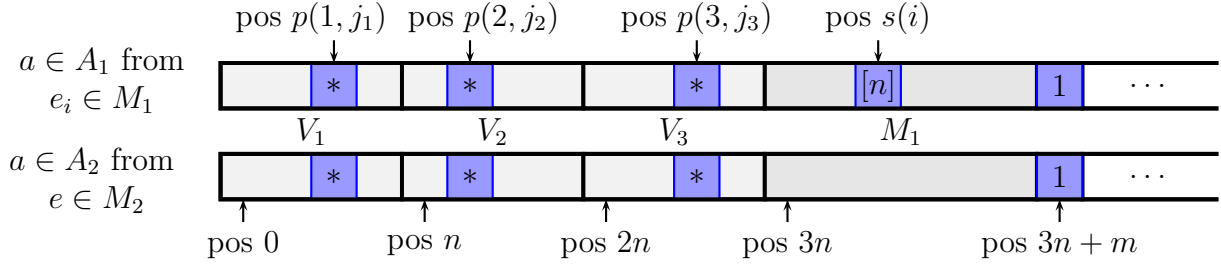
\begin{figure}
  \begin{center}
    \psset{unit=0.6cm}
    \begin{pspicture}(0,-1)(20,4)
      \psline[linewidth=1.5pt](22,1)(19,1)(19,0)(22,0) \rput[c](20.5,0.5){$\dots$}
      \rput[c](0,2){\psline[linewidth=1.5pt](22,1)(19,1)(19,0)(22,0)\rput[c](20.5,0.5){$\dots$}}
      \multido{\N=0+4}{3}{
        \drawRect{fillstyle=solid,fillcolor=black!5!white,linewidth=1.5pt}{\N}{2}{4}{1}
        \drawRect{fillstyle=solid,fillcolor=black!5!white,linewidth=1.5pt}{\N}{0}{4}{1}
      }
      \drawRect{fillstyle=solid,fillcolor=black!10!white,linewidth=1.5pt}{12}{2}{6}{1}
      \drawRect{fillstyle=solid,fillcolor=black!10!white,linewidth=1.5pt}{12}{0}{6}{1}
      \drawRect{fillstyle=solid,fillcolor=black!15!white,linewidth=1.5pt}{18}{2}{1}{1}
      \drawRect{fillstyle=solid,fillcolor=black!15!white,linewidth=1.5pt}{18}{0}{1}{1}
      \rput[r](-0.5,2.6){$\begin{array}{c} a \in A_1\textrm{ from} \\ e_i \in M_1 \end{array}$}
      \rput[r](-0.5,0.4){$\begin{array}{c} a \in A_2\textrm{ from } \\ e \in M_2 \end{array}$}
      \psline{->}(0.5,-0.5)(0.5,0)\rput[c](0.5,-1){pos $0$}
      \rput[c](4,0){\psline{->}(0.5,-0.5)(0.5,0)\rput[c](0.5,-1){pos $n$}}
      \rput[c](8,0){\psline{->}(0.5,-0.5)(0.5,0)\rput[c](0.5,-1){pos $2n$}}
      \rput[c](12,0){\psline{->}(0.5,-0.5)(0.5,0)\rput[c](0.5,-1){pos $3n$}}
      \rput[c](18,0){\psline{->}(0.5,-0.5)(0.5,0)\rput[c](0.5,-1){pos $3n+m$}}
      \rput[c](2,1.5){$V_1$}
      \rput[c](6,1.5){$V_2$}
      \rput[c](10,1.5){$V_3$}
      \rput[c](15,1.5){$M_1$}
      \multido{\N=0+2}{2}{
        \drawRect{fillstyle=solid,fillcolor=blue!40!white,linecolor=blue}{2}{\N}{1}{1}
        \drawRect{fillstyle=solid,fillcolor=blue!40!white,linecolor=blue}{5}{\N}{1}{1}
        \drawRect{fillstyle=solid,fillcolor=blue!40!white,linecolor=blue}{10}{\N}{1}{1}
        \drawRect{fillstyle=solid,fillcolor=blue!40!white,linecolor=blue}{18}{\N}{1}{1}
        \rput[c](0,\N){\rput[c](2.5,0.5){$*$}}
        \rput[c](0,\N){\rput[c](5.5,0.5){$*$}}
        \rput[c](0,\N){\rput[c](10.5,0.5){$*$}}
        \rput[c](0,\N){\rput[c](18.5,0.5){$1$}}
      }
      \drawRect{fillstyle=solid,fillcolor=blue!40!white,linecolor=blue}{14}{2}{1}{1}
      \rput[c](14.5,2.5){$[n]$}
      \multido{\N=0+2}{2}{\rput[c](0,\N){\psline[linewidth=1.5pt](22,1)(0,1)(0,0)(22,0)}}
      \rput[c](2,0){\psline{->}(0.5,3.5)(0.5,3.0)\rput[c](0.0,4.0){pos $p(1,j_1)$}}
      \rput[c](5,0){\psline{->}(0.5,3.5)(0.5,3.0)\rput[c](0.8,4.0){pos $p(2,j_2)$}}
      \rput[c](10,0){\psline{->}(0.5,3.5)(0.5,3.0)\rput[c](0.5,4.0){pos $p(3,j_3)$}}
      \rput[c](14,0){\psline{->}(0.5,3.5)(0.5,3.0)\rput[c](0.5,4.0){pos $s(i)$}}
    \end{pspicture}
    \caption{$b$-ary representation of number $a$ arising from edge $e_i = (v_{1,j_1},v_{2,j_2},v_{3,j_3}) \in M_1$ (top) and $e = (v_{1,j_1},v_{2,j_2},v_{3,j_3})\in M_2$ (bottom). Here $*$ means any number in $\{ 0,\ldots,n\}$.\label{fig:bAryRepOfA}}
  \end{center}
  \end{figure}
  We will prove that $(M_1,M_2) \in {\tt Q3DM} \Leftrightarrow (A,k) \in \overline{\texttt{FROB}}$. Writing out the definitions this means we prove 
  \begin{equation} \label{eq:ReductionEquivalence}
\left( \begin{array}{c} \forall S_1 \subseteq M_1 \; \exists S_2 \subseteq M_2:\\ S_1 \cup S_2\textrm{ is perfect matching}\end{array} \right)  \Longleftrightarrow  \left( \begin{array}{c} \textrm{all integers} \geq k \\ \textrm{are }A\textrm{-representable} \end{array} \right)
  \end{equation}
  {\bf Claim I.} \emph{``$\Rightarrow$'' in \eqref{eq:ReductionEquivalence}.} \\
  {\bf Proof of Claim I.} We assume that the LHS of \eqref{eq:ReductionEquivalence} holds. Fix an integer $x \geq k$. We need to prove that $x$ is $A$-representable. We observe that $A_{\min} := \min\{ a \mid a \in A \} = b^{3n+m}$, attained by setting $d_1 = d_2 = d_3 = 0$ for a number arising from any edge in $M_2$.
  It suffices to prove the claim for $x$ with $nb^{3n+m} = k \leq x \leq k+A_{\min} -1 = b^{3n+m+1}-1$ since
  we can then add integral multiples of $A_{\min}$ to cover all integers in $[k,\infty)$.
  Hence we can write $x$ in its $b$-ary representation as
  \[
x = n b^{3n+m} + \sum_{\ell=0}^{3n+m-1} z_{\ell}b^{\ell} \quad \textrm{where} \quad z_{\ell} \in \{0,\ldots,n\}
\]
We set $S_1 := \{ e_i \in M_1 \mid z_{s(i)} > 0\}$. Then by assumption,
there is a set $S_2 \subseteq M_2$ so that $S_1 \cup S_2$ is a perfect matching. In particular $|S_1| + |S_2|=n$.
For each edge $e \in S_1 \cup S_2$ we will choose a single number $a_e \in A$
so that  $x = \sum_{e \in S_1 \cup S_2} a_e$. The construction of the numbers
is as follows:
\begin{itemize}
\item Case $e = e_i \in S_1$. We write the edge as $e = (v_{1,j_1},v_{2,j_2},v_{3,j_3})$. Then we set
  \[
   a_{e_i} := b^{3n+m} + z_{s(i)} b^{s(i)} + \sum_{r=1}^3 z_{p(r,j_r)} b^{p(r,j_r)}
  \]
\item Case $e \in S_2$. Using the same notation as above, we make the same choice minus the middle term, i.e.
  \[
   a_{e} := b^{3n+m} + \sum_{r=1}^3 z_{p(r,j_r)} b^{p(r,j_r)}    
  \]
\end{itemize}
Since $S_1 \cup S_2$ is a perfect matching, each digit $\ell \in \{ 0,\ldots,3n-1\}$ (belonging to  $V$) receives one designated contribution which matches $z_{\ell}$. Positions $s(i)$ with $e_i \in S_1$ receive $z_{s(i)} \in [n]$; if $e_i \notin S_1$ then the position is not in the support of any number.
Moreover we have selected exactly $n$ numbers and so the $(3n+m)$-th digit receives $n$. No carry-over occurs.
Hence indeed $x = \sum_{e \in S_1 \cup S_2} a_e$.  \qed

{\bf Claim II.} \emph{``$\Leftarrow$'' in \eqref{eq:ReductionEquivalence}.} \\
{\bf Proof of Claim II.}
For the converse, we assume that all integers $\geq k$ are $A$-representable. Fix any set $S_1 \subseteq M_1$; we need to show that there is a set $S_2 \subseteq M_2$ so that $S_1 \cup S_2$ is a perfect matching. Consider the number
 \begin{equation} \label{eq:xRep1}
  x := nb^{3n+m} + n\sum_{\ell=0}^{3n-1} b^{\ell} + n\sum_{e_i \in S_1} b^{s(i)} 
\end{equation}
By assumption we can write
\begin{equation} \label{eq:xRep2}
  x = \sum_{a \in A} y_a a
\end{equation}
where $y \in \setZ_{\geq 0}^A$. For every number $a \in A$ we choose a witness denoted by  $e_a \in M_1 \cup M_2$ which is used to generate the number.
Since $x < (n+1)b^{3n+m}$ and $A_{\min} = b^{3n+m}$ we can infer that $\|y\|_1 \leq n$.
We claim that when adding up $y_a \cdot a$ over $a \in A$ in base $b$, no carry-over occurs. To see this, consider the restriction to the digits in $V$ (which are the lower order digits). Each number $a \in A$ contributes at most $y_a \cdot n$ to at most $3$ digits in $V$ while $\|y\|_1 \leq n$. At the same time, from~\eqref{eq:xRep1} we know that each of the $|V| = 3n$ digits receives a contribution of $n$. This may add up,
but it leaves no room for slack. In particular we necessarily have $\|y\|_1 = n$ and no carry-over can
occur because then the contribution of digit $\ell$ to digit $\ell' > \ell$ would inflate by a factor of $b^{-(\ell'-\ell)} < 1$. By the same ``no slack'' argument we know that $y_a \in \{ 0,1\}$ for all $a \in A$ and we also know that for distinct $a,a' \in \textrm{supp}(y)$ one has $e_a \neq e_{a'}$.
We define the edge set  $T := \{ e_a : a \in A\textrm{ and }y_a = 1\}$. From the discussion above we already know that $T$ is a perfect matching. Since every edge contributes at most one number to $\textrm{supp}(y)$ and $\|y\|_{\infty} \leq 1$, we know that no carry-over can occur in the indices belonging to $M_1$ either. 
In particular by \eqref{eq:xRep1} this implies that  $T \cap M_1 = S_1$. Then $S_2 := T \setminus S_1$ satisfies the claim.  
\end{proof}

\bibliographystyle{alphaurl} 
\bibliography{frobenius}

@misc{matsubara2016computationalcomplexityfrobeniusproblem,
      title={The Computational Complexity of the Frobenius Problem}, 
      author={Shunichi Matsubara},
      year={2016},
      eprint={1602.05657},
      archivePrefix={arXiv},
      primaryClass={cs.CC},
      url={https://arxiv.org/abs/1602.05657}, 
}

@inproceedings{KannanFSTTCS1989,
  author       = {Ravi Kannan},
  editor       = {C. E. Veni Madhavan},
  title        = {The Frobenius Problem},
  booktitle    = {Foundations of Software Technology and Theoretical Computer Science,
                  Ninth Conference, Bangalore, India, December 19-21, 1989, Proceedings},
  series       = {Lecture Notes in Computer Science},
  volume       = {405},
  pages        = {242--251},
  publisher    = {Springer},
  year         = {1989},
  url          = {https://doi.org/10.1007/3-540-52048-1\_47},
  doi          = {10.1007/3-540-52048-1\_47},
  bibsource    = {dblp computer science bibliography, https://dblp.org}
}

@article{RamrezAlfonsn1996ComplexityOT,
  title={Complexity of the Frobenius problem},
  author={Jorge L. Ram{\'i}rez-Alfons{\'i}n},
  journal={Combinatorica},
  year={1996},
  volume={16},
  pages={143-147},
  url={https://api.semanticscholar.org/CorpusID:118624008}
}

@book{AroraBarak2009,
  author       = {Sanjeev Arora and
                  Boaz Barak},
  title        = {Computational Complexity - {A} Modern Approach},
  publisher    = {Cambridge University Press},
  year         = {2009},
  url          = {http://www.cambridge.org/catalogue/catalogue.asp?isbn=9780521424264},
  isbn         = {978-0-521-42426-4},
  bibsource    = {dblp computer science bibliography, https://dblp.org}
}

@ARTICLE{McLoughlinHardness1984,
  author={McLoughlin, A.},
  journal={IEEE Transactions on Information Theory}, 
  title={The complexity of computing the covering radius of a code}, 
  year={1984},
  volume={30},
  number={6},
  pages={800-804},
  doi={10.1109/TIT.1984.1056978}}

@article{Brauer1942,
 ISSN = {00029327, 10806377},
 URL = {http://www.jstor.org/stable/2371684},
 author = {Alfred Brauer},
 journal = {American Journal of Mathematics},
 number = {1},
 pages = {299--312},
 publisher = {Johns Hopkins University Press},
 title = {On a Problem of Partitions},
 urldate = {2026-08-28},
 volume = {64},
 year = {1942}
}

@article{Beihoffer_Hendry_Nijenhuis_Wagon_2005,
title={Faster Algorithms for Frobenius Numbers},
volume={12}, url={https://www.combinatorics.org/ojs/index.php/eljc/article/view/v12i1r27},
DOI={10.37236/1924},
number={1},
journal={The Electronic Journal of Combinatorics},
author={Beihoffer, Dale and Hendry, Jemimah and Nijenhuis, Albert and Wagon, Stan},
year={2005},
month={Jun.},
pages={R27}
}

@article{KannanCombinatorica92,
  author       = {Ravi Kannan},
  title        = {Lattice translates of a polytope and the Frobenius problem},
  journal      = {Comb.},
  volume       = {12},
  number       = {2},
  pages        = {161--177},
  year         = {1992},
  url          = {https://doi.org/10.1007/BF01204720},
  doi          = {10.1007/BF01204720},
  bibsource    = {dblp computer science bibliography, https://dblp.org}
}

\end{document}